\documentclass[10pt,conference]{IEEEtran}
\IEEEoverridecommandlockouts
\usepackage{color}
\usepackage{epsf}
\usepackage{times}
\usepackage{epsfig}
\usepackage{epstopdf}
\usepackage{graphicx}
\usepackage{graphics}
\usepackage[small]{caption}
\usepackage{url}
\usepackage{amsmath}
\usepackage{amssymb}
\usepackage{amsxtra}
\usepackage{algorithmic}
\usepackage{setspace}
\usepackage{algorithm}
\usepackage{enumerate}
\usepackage{multirow}
\usepackage{citesort}
\usepackage{enumerate}
\usepackage{amssymb}
\usepackage{rotating}
\usepackage{colortbl}
\usepackage{type1cm}
\usepackage{eso-pic}
\usepackage{color}
\makeatletter
\AddToShipoutPicture{%
            \setlength{\@tempdimb}{.5\paperwidth}%
            \setlength{\@tempdimc}{.5\paperheight}%
            \setlength{\unitlength}{1pt}%
            \put(\strip@pt\@tempdimb,\strip@pt\@tempdimc){%
        \makebox(0,0){\rotatebox{45}{\textcolor[gray]{0.75}%
        {\fontsize{1.2cm}{1.2cm}\selectfont{ICC 2013 (Accepted)}}}}%
            }%
} \makeatother

\newtheorem{theorem}{\bf Theorem}

\newtheorem{definition}{\bf Definition}

\newcommand{\INDSTATE}[1][1]{\STATE\hspace{#1\algorithmicindent}}

\ifCLASSINFOpdf

\else

\fi

\begin{document}

\title{Prioritizing Consumers in Smart Grid: Energy Management Using Game Theory}

\author{\IEEEauthorblockN{Wayes Tushar\IEEEauthorrefmark{1}\IEEEauthorrefmark{4},
Jian A. Zhang\IEEEauthorrefmark{3}, David B.
Smith\IEEEauthorrefmark{4}\IEEEauthorrefmark{1}, Sylvie
Thiebaux\IEEEauthorrefmark{4}\IEEEauthorrefmark{1} and H.~Vincent~Poor\IEEEauthorrefmark{2}}
\IEEEauthorblockA{\IEEEauthorrefmark{1}College of Engineering and
Computer Science,~Australian National University, ACT, Australia. Email:
wayes.tushar@anu.edu.au\\\IEEEauthorrefmark{3}CSIRO
ICT Center, Marsfield, NSW, Australia. Email: andrew.zhang@csiro.au\\
\IEEEauthorrefmark{4}National ICT Australia
(NICTA),~Canberra,~ACT,~Australia. Email: \{david.smith,
sylvie.thiebaux\}@nicta.com.au\\\IEEEauthorrefmark{2}School of
Engineering and Applied Science, Princeton University, Princeton,
NJ, USA. Email: poor@princeton.edu}
\thanks{This work is supported by NICTA. NICTA is funded by the Australian Government as
represented by the Department of Broadband, Communications and the
Digital Economy and the Australian Research Council through the ICT
Centre of Excellence program.This work is also supported in part by the U. S. Air Force Office of Scientific Research under MURI Grant FA9550-09-1-0643 .}
}
\date{}
\maketitle
\def\baselinestretch{.916}
\begin{abstract}
This paper explores an idea of demand-supply balance for smart grids
in which consumers are expected to play a significant role. The main
objective is to \emph{motivate} the consumer, by \emph{maximizing}
their benefit both as a seller and a buyer, to trade their surplus
energy with the grid so as to balance the demand at the peak hour.
To that end, a Stackelberg game is proposed to capture the
interactions between the grid and consumers, and it is
shown analytically that optimal energy trading parameters that maximize
customers' utilities are obtained at the solution of the game. A
novel distributed algorithm is proposed to reach the optimal
solution of the game, and numerical examples are used to assess the
properties and effectiveness of the proposed approach.
\end{abstract}

\IEEEpeerreviewmaketitle
\section{Introduction}\label{introduction}
The smart grid (SG) is envisioned to be a large-scale next
generation cyber-physical system that will improve the efficiency,
reliability, and robustness of future power and energy grids by
integrating the consumers as one of its key management
components~\cite{2011IEEE-JCST_Fang}, and thus, achieve a system
which is clean, safe, reliable, resilient and sustainable. This
heterogeneous network will motivate the adoption of advanced
technologies that will increase the participation of its consumers
to overcome various technical challenges at different levels of
demand-supply balance~\cite{2011IEEE-JCST_Fang}.

In this respect, game theory, which is an analytical framework to
capture the complex interactions among rational
players~\cite{1999Book_Dynamicgame-Basar} is studied in this paper
to model an energy trading scheme for the SG. The model uses the
two-way communication facility of the SG~\cite{2010IEEE-JTSG_Rad},
and inspires the customers to \emph{spontaneously} take part in
\emph{supplying} their surplus energy (SE) to the grid so as to
assist the power grid (PG) in balancing the excess energy demand at
the peak hour. This voluntary participation of consumers in energy
trading is very important in the context of SG because of its
ability to greatly enhance the SG's reliability, and thus, improve
the social benefit of the electricity
market~\cite{2010-JEnergy_Walawalkar}. We use the framework of a
Stackelberg game~\cite{1999Book_Dynamicgame-Basar} for this model in
which the PG is considered as the leader and energy users (EUs) are
the followers. Here, on the one hand, the PG decides on the total
amount of energy it wants to buy, and also on the price per unit of
energy it needs to pay to each EU. On the other hand, the EU decides
on its amount of energy to be sold to the PG in response to the
price offered to it.

We note that energy management in the context of SG has been
receiving considerable attention recently. For example, energy
management for SGs in a vehicle-to-grid (V2G) scenario have been
studied in
\cite{2011IEEE-JCST_Fang,2010IEEE-JTSG_Rad,2012IEEE-JTSG_Wu}
and the references therein, whereby the application of game theory
for demand-supply balance in SGs can be found
in~\cite{2010IEEE-JTSG_Rad} and \cite{2012IEEE-JTSG_Wu}. However,
little has been done in prioritizing the consumers' benefit in
management modeling where the main priority of the energy management
scheme is to benefit the consumers. We stress that consumers are the
core element of the evolution of SG as explained in
\cite{2011IEEE-ISGT_Liu}, and hence, their benefit is one of the
most important concerns of any demand-supply modeling scheme. In
this respect, we propose an energy management scheme that
prioritizes the consumers in the SG and balances demand with supply
at peak hours. We first formulate a noncooperative \emph{Stackelberg}
game (NSG) to study the interactions between the PG and EUs in the
SG, and show that the optimal demand-supply balance can be achieved
at the solution of the game; then we analyze properties of the game
in terms of existence and optimality, and it is shown that the game
possesses a socially optimal solution; finally, we propose a
distributed algorithm to reach the solution of the game, and the
effectiveness of the proposed scheme is demonstrated via numerical
experiments.
\section{System Model}\label{system-model}
Consider an SG network that consists of a single PG and a number of
EUs. The set of EUs is $\mathcal{N}$, where $|\mathcal{N}|=N$.
Here, the PG refers to the main electricity grid which is servicing a
group of customers at peak hours of the day (i.e., $12$ pm to $4$
pm), and each EU $i\in\mathcal{N}$ is a group of similar idle energy
users~\cite{2011IEEE-JCST_Fang}, connected via an
aggregator~\cite{2010IEEE-JTSG_Sekyung}, such as smart homes,
electric vehicles, wind mills, solar panels and bio-gas plants that
have some SE for sale after regular usage. It is assumed that the PG
can communicate with the EUs through smart meters via an appropriate
communication protocol~\cite{2010IEEE-JTSG_Rad}. \vspace{-0.05cm}

Due to frequent change of energy state in the grid, energy
management in the SG needs to be carried out
frequently~\cite{2012IEEE-JTSG_Wu}, and therefore, the total peak
hour duration can be divided into multiple time
slots~\cite{2012IEEE-JTSG_Wu}. As energy demands by the customers
are very high during the peak hours, the PG may be unable to balance some
of the demands from its own generation in some of these time slots.
Meanwhile, the PG needs to buy energy from alternative energy
sources such as idle EUs who have SE, and may agree to sell it to
the PG with appropriate incentives. For the rest of the paper, we
will concentrate on the energy management in a single time slot. It
is assumed that the energy deficiency of the PG, $E_{\text{PG}}$, at
any time slot is fixed although the deficiency may vary from one
time slot to the next. However, as the required energy by the PG is
fixed during a time slot, the PG would not be interested in buying
more energy than $E_{\text{PG}}$  to keep its cost at a minimum.
Thus, if each EU $i$ with SE $E_i$ provides the PG with energy
$x_i$, these quantities need to satisfy
\begin{eqnarray}
\sum_i x_i\leq E_{\text{PG}};~x_i\leq E_i,~\forall
i.\label{constraint-1}
\end{eqnarray}
\vspace{-0.05cm}
To buy an offered amount of energy $x_i$, the PG pays a price $p_i$ per unit
of energy to EU $i$ as an incentive. However, the PG may need to pay
different incentives to different EUs due to their different amounts of SE. For
instance, a lower incentive may not affect the intended revenue of
an EU with higher SE as it can sell more, but could severely affect
the revenue of EUs with smaller amounts of SE. Moreover, the PG may also want to
minimize its total cost of purchase as it would further enable the
PG to sell this energy at a cheaper rate to its customer. This
would facilitate the trading of energy between the PG and the EUs in
the network rather than establishing more expensive generators or
bulk capacitors to meet any excess demand, and also, the cheaper
rate would benefit the consumers who buy the energy from the PG. To this
end, we assume that the PG estimates a total price $P_r$ per unit of
energy, analogous to the \emph{total cost per unit production} in
economics~\cite{2010-Book_Farris}, in each time slot using a real
time price estimation technique as proposed in
\cite{2008IEEE-JTPS_Yun}. The PG uses this $P_r$ to optimize the
price $p_i$ it will pay to each EU $i$ to order to minimize its total cost while
maintaining the constraint
\begin{eqnarray}
\sum_i p_i=P_r;~p^{\text{min}}\leq p_i\leq
p^{\text{max}}.\label{constraint-2}
\end{eqnarray}
The equality constraint in \eqref{constraint-2}  establishes that
the announced total price per unit energy must be paid to all EUs,
and thus motivates the EUs to take part in energy trading with the
PG. Here, $p^{\text{min}}$ is the minimum price that the PG needs to
pay EU $i$ to incentivize it to trade energy, and $p^{\text{max}}$ is the
maximum price that the PG can pay. Although $P_r$ is fixed, $p_i$
can be different for different $i$ based on $x_i$.
\section{Stackelberg Game and Its Properties}\label{game-formulation}
In a consumer-prioritized SG, the beneficiaries of the energy
management scheme are the consumers in the
network~\cite{2011IEEE-ISGT_Liu}. In this regard, we propose an NSG,
in which on the one hand, the objective of each EU $i$ is to
voluntarily sell an amount of energy $x_i$ to the PG based on $E_i$ and the
offered price $p_i$ so as to maximize its own utility. On the other
hand, the PG wants to minimize its total cost of purchase by
optimizing $p_i$ for different $i$ as explained
in Section~\ref{system-model}. To this end, we now define the
objective functions of the leader and followers of the game.
\subsection{Objective functions of the EUs and the PG}
The considered utility function of EU $i$, $U(x_i,E_i,p_i)$, is based on a
linearly decreasing marginal benefit, which is appropriate for
energy users~\cite{2010IEEE-CSmartGridComm_Samadi}. In addition, the
utility function is also assumed to possess the following
properties: i) the utility of EU $i$ increases with the amount of
SE $E_i$, i.e., $\frac{\delta U}{\delta E_i}>0$; and ii) the
utility of an EU increases as $p_i$ increases, i.e., $\frac{\delta
U}{\delta p_i}>0$. To meet the above properties, in this work we
consider the following utility function for EU $i$:
\begin{equation}
U(x_i,E_i,p_i)=E_ix_i-\frac{1}{2}{x_i}^2+p_ix_i.\label{utility-EU-single}
\end{equation}
From \eqref{utility-EU-single}, we note that the addition of the
quantity $\sum_{j\neq i}E_jx_j-\frac{1}{2}{x_j}^2+p_jx_j$ to $U$ in
\eqref{utility-EU-single}  does not affect the solution.
Consequently, all the EUs equivalently maximize
\begin{equation}
\hat{U}(\mathbf{x},\mathbf{E},\mathbf{p})=\sum_{i=1}^N
U(x_i,E_i,p_i),\label{utility-joint}
\end{equation}
subject to the constraint $\sum_ix_i\leq E_{\text{PG}},~\forall i$.
Here, $\mathbf{x}=[x_1, x_2, \hdots, x_N]^T$, $\mathbf{E}=[E_1, E_2,
\hdots, E_N]^T$ and $\mathbf{p}=[p_1, p_2, \hdots, p_N]^T$.

On the other hand, the PG's target is to decide on its price $p_i$
based on the offered energy $x_i$ by EU $i$ so as to minimize the
total cost of purchase $\tilde{J}(\mathbf{p})=\sum_i J(p_i)$. Here,
$J(p_i)$ is the individual cost including the cost of purchasing
with price $p_i$ and other associated costs. For $J(p_i)$, we
consider a convex cost function
\begin{equation}
J(p_i)=x_i{p_i}^2+a_ip_i+b_i;~a_i,b_i>0;~\forall i\in\mathcal{N},\label{cost-function}
\end{equation}
which is comparable to the practical cost function of some utility
companies~\cite{2010IEEE-JTSG_Rad}. In \eqref{cost-function}, the
first term captures the cost of purchasing energy, whereas the
associated costs\footnote{Examples include transmission and
artificial tariff costs.} are reflected in the last two terms. In
\eqref{cost-function}, $J$ possesses the following properties: i)
$J$ increases with the increase of $p_i$. That is
$J(\hat{p}_i)<J(\acute{p}_i),~\forall \hat{p}_i<\acute{p}_i$; and
ii) $J$ is strictly convex. Thus, for $\hat{p}_i,\acute{p}_i\geq 0,
\forall i$ and any real number $0<\gamma<1$,
$J(\gamma\hat{p}_i+(1-\gamma)\acute{p}_i)<\gamma
J(\hat{p}_i)+(1-\gamma)J(\acute{p}_i)$. The objective of the PG is
to minimize its total cost, and thus, the net cost function of the
PG is
\begin{eqnarray}
\tilde{J}(\mathbf{p})=\sum_{i=1}^{N}
\left(x_i{p_i}^2+a_ip_i+b_i\right);\label{net-cost}\\\text{s.t.},
\sum_ip_i=P_r;~p^{\text{min}}\leq p_i\leq p^{\text{max}}~\forall
i.\nonumber
\end{eqnarray}
As EUs are owned by individual consumers, the PG cannot directly
control their decision making processes, and hence, a decentralized
control mechanism is required for both EUs and the PG to decide on
$x_i$ and $p_i$ by optimizing their respective utility and cost
functions in \eqref{utility-joint} and \eqref{net-cost}.
\subsection{Formulation of the game}
The PG and EUs interact with each other to decide on their energy
trading parameters $p_i$ and $x_i$, and we propose an NSG\footnote{A similar form of game for charging of electric vehicles in smart grid was used in \cite{2012IEEE-JTSG_Tushar}.} $\Lambda$ to capture this
interaction. In this game, the PG is the leader who decides on the
price $p_i$ for the amount of energy offered by the EU $i$, and each
EU $i$ is a follower who agrees on $x_i$ to be offered to the PG in
response to $p_i$ by playing a generalized Nash game
(GNG)~\cite{2011-CWWWC_Arganda} with other followers in the SG.
Hence, the proposed NSG can be defined by its strategic form as
\begin{equation}
\Lambda=\lbrace(\mathcal{N}\cup\{\text{PG}\}),\{\mathbf{X}_i\}_{i\in\mathcal{N}},\hat{U},\tilde{J},\mathbf{p}\rbrace,\label{NSG}
\end{equation}
where $(\mathcal{N}\cup\{\text{PG}\})$ is the set of all players in
the game, and $\{\mathbf{X}_i\}_{i\in\mathcal{N}}$ is the vector of
strategies of EU $i$ satisfying \eqref{constraint-1}.

In this game, an EU's decision is affected by the strategies of
other EUs through \eqref{constraint-1}, and thus, the GNG amongst
EUs, to decide on $x_i~\forall i$, is a jointly convex generalized
Nash equilibrium problem (GNEP)~\cite{2011-CWWWC_Arganda}, whose
solution is a generalized Nash equilibrium (GNE). The proposed game
is played in two steps. In the first step, the game is initiated
with the announcement of the price $p_i=p,~\forall i$ satisfying
\eqref{constraint-2} by the PG, upon which the EUs play a GNEP to
decide on the GNE energy set $\mathbf{x}$ they wish to sell to the
PG. The PG receives the offered energy $\mathbf{x}$ and thereby
obtains some insight into the capacity\footnote{For a similar price
$p$, each EU receives a similar incentive, and thus their offered
energies reflect their capacities of supply.} of each EU $i$. In the
second step, the PG optimizes its price vector to
$\mathbf{p^*}=[{p_1}^*,\hdots,{p_N}^*]^T$, by a constraint
optimization technique for the offered $\mathbf{x}$. Thereafter, the
EUs again decide on their GNE energy vector
$\mathbf{x}^*=[{x_1}^*,\hdots,{x_N}^*]^T$ in response to
$\mathbf{p}^*$ and the proposed NSG reaches its noncooperative
Stackelberg equilibrium (NSE).
\begin{definition}
\label{definition-1} Consider the NSG $\Lambda$ in \eqref{NSG}, in
which $\hat{U}$ and $\tilde{J}$ are defined by \eqref{utility-joint} and
\eqref{net-cost} respectively. A set of strategies
$\left(\mathbf{x}^*, \mathbf{p}^*\right)$ constitutes the NSE of the
game if and only if the strategy set satisfies the following set of
inequalities:
\begin{align}
\hat{U}({x_i}^*,{\mathbf{x}_{-i}}^*,\mathbf{E},{\mathbf{p}}^*)\geq
\hat{U}({x_i},{\mathbf{x}_{-i}}^*,\mathbf{E},{\mathbf{p}}^*),\label{SE_cond-1}\\\forall
x_i\in\mathbf{x},~\forall i\in\mathcal{N},~\sum_ix_i\leq
E_{\text{PG}},\nonumber
\end{align}
and
\begin{align}
\tilde{J}({p_i}^*,\mathbf{p}_{-i}^*)\leq \tilde{J}({p_i},\mathbf{p}_{-i}^*),\label{SE_cond-2}\\\forall i\in\mathcal{N},~\forall p_i\in\mathbf{p},~p^{\text{min}}\leq p_i\leq p^{\text{max}}\nonumber.
\end{align}
Here, $\mathbf{x}_{-i}$ is the energy vector of all EUs in the set
$\mathcal{N}$ except EU $i$, and similarly, $\mathbf{p}_{-i}$ is the
price vector set by the PG for all EUs in $\mathcal{N}$ excluding
$i$.
\end{definition}

Therefore, at the NSE, neither any EU nor the PG can increase its utility by deviating from its NSE strategy while all other players in the SG are playing their NSE strategies.
\subsection{Existence and optimality}
In this section, we investigate the existence of a \emph{socially
optimal} solution which is beneficial for all the consumers in the
SG, and thus, suitable for the proposed consumer-prioritized energy
management scheme. In this regard, first we note that the proposed
NSG reaches the NSE as soon as all EUs agree on a GNE energy vector
$\mathbf{x}^*$ in response to the optimized price vector
$\mathbf{p}^*$ set by the PG. Due to the strict convexity of
\eqref{net-cost}, there always exists a unique
solution~\cite{2004-Book_BOYD} for the price vector $\mathbf{p}^*$,
and thus, the existence of a socially optimal GNE of the
\emph{followers' GNEP} would \emph{guarantee} the existence of a
socially optimal NSE of the \emph{proposed NSG}.
\begin{theorem}
\label{theorem-1}
There exists a socially optimal GNE solution for the GNEP amongst the EUs in response to the price set by the PG.
\end{theorem}
\begin{proof}
First, we note that the proposed GNEP is a jointly convex GNEP due
to the coupled constraint \eqref{constraint-1}, and hence, the GNEP
can be formulated as a variational inequality (VI) problem
VI($\mathbf{X,F}$)~\cite{2011-CWWWC_Arganda}, which can be used
to determine a vector $\mathbf{x}^*\in\mathbf{X}\in\mathbb{R}^n$
such that $\langle \mathbf{F}(\mathbf{x}^*),\mathbf{x-x^*}\rangle
\geq 0,\forall x\in\mathbf{X}$. Here, $\mathbf{X}$ is the vector of
strategies of all EUs satisfying \eqref{constraint-1} and,
from~\cite{2011-CWWWC_Arganda}
\begin{align}
\mathbf{F}=-\left(\nabla_x U(x_i,E_i,p_i)\right)_{i=1}^{N}.
\end{align}
The solution of a VI is a variational equilibrium (VE), which is the
socially optimal solution among other
GNEs~\cite{2011-CWWWC_Arganda}. To this end, our main concern it to
study the existence of the VE of the GNEP, and thus, check the
optimality of the solution. Hereinafter, we will use GNEP and VI
interchangeably, and VE to refer to both the VE and GNE of the GNEP. Now,
the pseudo-gradient of the joint utility function
$\hat{U}$ in \eqref{utility-joint} is~\cite{2011-CWWWC_Arganda}
\begin{equation}
\mathbf{F}=\left[\begin{array}{c} x_1-E_1-p_1\\x_2-E_2-p_2\\\vdots\\
x_N-E_N-p_N\end{array}\right],\label{eqn-pseudo-grad}
\end{equation}
and the Jacobean of \eqref{eqn-pseudo-grad}, $\mathbf{JF}$, is an
identity matrix. Hence, $\mathbf{JF}$ is positive definite on
$\mathbf{X}$, and consequently, $\mathbf{F}$ is strictly monotone.
Therefore, VI$(\mathbf{X,F})$ admits a unique
VE~\cite{2007-J4OR_Facchinei}. Moreover, due to the joint convexity
of the proposed GNEP, the unique VE is also the global unique
maximizer of~\eqref{utility-joint}~\cite{2007-J4OR_Facchinei}, and
thus, the existence of a \emph{socially optimal} VE of the proposed
GNEP is proved.
\end{proof}
From Theorem~\ref{theorem-1}, we can further conclude that the
proposed NSG possesses a socially optimal NSE.
\subsection{Algorithm}
In this section, we propose a distributed algorithm for the
considered NSG to reach the socially optimal NSE. The algorithm is
implemented based on the fact that $\mathbf{F}$ is strongly
monotone, and thus, the slack variables $\mu_i=E_i-x_i+p_i,~\forall
i\in\mathcal{N}$ all possess the same value $\mu_i=\mu$ at the
VE~\cite{2007-J4OR_Facchinei}. We assume that there is limited
communication between the PG and EUs via the two-way communications
of the SG, and thus, the PG can inform\footnote{This can be done, for example, by using a
single bit.} the EU $i$ if its offered energy $x_i$ is beyond the VE
in response to $p_i^*$ (by checking whether $\mu_i=\mu,~\forall i$).
We use a hyperplane projection method for the VI, which is the
simplest solution method for a monotone VI~\cite{2011-CWWWC_Arganda},
and its solution is always guaranteed to converge to a non-empty
VE~\cite{2011-CWWWC_Arganda}. Since the PG's optimization problem is
strictly convex and has a unique solution, consequently, the
proposed algorithm also guarantees that the NSG will reach its socially
optimal NSE.

The algorithm is executed in two steps as shown in
Algorithm~\ref{algorithm-1}. As a hyperplane projection method, we
use the S-S method~\cite{2011-CWWWC_Arganda,1999-JSIAMJCP_Solodov},
which is based on a geometrical interpretation and uses two
projections per iteration. Let $x^l$ be the current approximation
of the solution of VI$(\mathbf{X,F})$. In the S-S method, first the
projection $r(x^l)=\text{Prj}_{\mathbf{X}}\left(x^l-F(x^l)\right)$
is computed\footnote{$\text{Prj}_{\mathbf{X}}(k)=\arg \min\lbrace
\|\omega-k\|,\omega\in\mathbf{X}\rbrace,~\omega\in\mathbb{R}^n$.},
and then a point $z^l$ is determined in the line segment between $x^l$ and
$r(x^l)$, using an Armijo-type search~\cite{2011-CWWWC_Arganda}.
Then, $x^l$ is separated from the other solution $x^*$ of the problem
via the hyperplane $\delta M^l=\lbrace x\in\mathbb{R}^n|\langle
F(z^l),x-z^l\rangle =0\rbrace$. Now, as soon as the hyperplane is
constructed, $x^{l+1}$ is computed in the next iteration onto the
feasible set of $\mathbf{X}$ with the hyperspace $\delta M^l=\lbrace
x\in\mathbb{R}^n|\langle F(z^l),x-z^l\rangle \leq 0\rbrace$ which
contains the solution set~\cite{1999-JSIAMJCP_Solodov}. The details
of the S-S method can be found in~\cite{1999-JSIAMJCP_Solodov}.
%
\begin{algorithm}[t]
\caption{Algorithm to reach a socially optimal NSE}
\begin{algorithmic}
\scriptsize %
\STATE \textbf{\underline{\emph{Step-1}}}%
\STATE\hspace{1mm}~(i)- The PG announces $E_{\text{PG}}$ and $P_r$.%
\STATE\hspace{1mm}~(ii)- EU $i$ calculates $p_i=\frac{P_r}{N}$,
and plays a GNEP to determine VE energy $x_i$, for $p_i$, using the S-S method~\cite{1999-JSIAMJCP_Solodov}.%
\INDSTATE\textbf{\emph{S-S~method:}}%
\INDSTATE a) At iteration $l$, EU $i\in\mathcal{N}$ computes the hyperplane projection $r(x_i^l)$ and updates $x_{i}^{l+1}=r(x_i^l)$.%
\INDSTATE b) The EU checks: if $r(x_i^l)=0$\\%
\hspace{0.5cm}\textbf{if} $r(x_i^l)=0$\\%
\hspace{1cm}a) The EU chooses the energy $x_l$ to offer to the PG.\\%
\hspace{0.5cm}\textbf{else}\\%
\hspace{1cm} a) the EU $i$ determines the hyperplane $z_i^l$ and the half space $M_i^l$ from\\%
\hspace{0.5cm} the projection.\\%
\hspace{1cm} b) the EU updates the amount $x_i^{l+1}$ from the projection of its previous\\%
\hspace{0.5cm} energy $x_i^l$ on to $X\cap M_i^l$ and chooses to submit to the PG.\\%
\hspace{0.5cm}\textbf{end if}%
\STATE (iii). The EU $i$ determine $\mu_i=E_i-x_i+p_i$ and submits it to the PG. %
\STATE (iv). The PG checks $\mu_i,~\forall i\in\mathcal{N}$. \\%
\hspace{0.5cm}\textbf{if} $\mu_1=\mu_2=...=\mu_N=\mu$\\%
\hspace{1cm} The PG determines the VE energy vector $\mathbf{x}$ of all the EUs in the network\\%
\hspace{0.5cm}  in response to $p_i=p,~\forall i\in\mathcal{N}$. \\%
\hspace{0.5cm}\textbf{else}\\%
\hspace{1cm} The PG directs the EUs to repeat (ii) and (iii).\\%
\hspace{0.5cm}\textbf{end if}%
\STATE\hspace{1mm}~(v)- EUs submit their offered VE energies for $p_i~\forall i\in\mathcal{N}$ to the PG.%
\STATE\textbf{\underline{\emph{Step-2}}}%
\STATE\hspace{1mm}~(vi)- The PG optimizes \eqref{net-cost} using a
standard convex optimization technique~\cite{2004-Book_BOYD}, and determines $p_i=p_i^*~\forall i\in\mathcal{N}$.%
\STATE\hspace{1mm} \textbf{The optimized price vector $\mathbf{p^*}$ is obtained}.%
\STATE\hspace{1mm}~(vii)- EU $i$ receives the optimized price $p_i^*$ as offered by the PG.%
\STATE\hspace{1mm}~(viii)- Each EU $i$ again plays a GNEP for the offered price $p_i^*$, following steps (ii),(iii), and (iv) so as to
determine the VE energy $e_i^*$ to supply to the PG in response to $\mathbf{p}^*$.
\STATE\hspace{1mm} \textbf{The VE energy vector $\mathbf{e^*}$ for
$\mathbf{p^*}$ is obtained}.%
\STATE\emph{\textbf{The NPG reaches the NSE.}}%
\end{algorithmic}
\label{algorithm-1}
\end{algorithm}
\section{Numerical Experiments}\label{numerical-simulation}
We simulate the proposed energy management scheme by considering a
number of different EUs, where each EU represents $20$ similar
energy consumers connected via an aggregator. The available SE of
each EU is assumed to be a uniformly distributed random variable in
the range of $\left[64,240\right]$ kilowatt hour (kWh), and thus
covers both the lowest battery capacity of a group of solar panels
(3.6 kWh per panel) and the highest battery
capacity of a wind turbine group (12.25 kWh per
turbine). The total price per unit energy is assumed to be $175$ US
cents\footnote{For $5$ EUs, the average price per unit energy is
$35$ cents/kWh~\cite{2012-solarchoice}.}, and
$(p^{\text{min}},p^{\text{max}})=(8.45, 175)$ cents, unless stated
otherwise. All the statistical results are averaged over random values of the EUs' capacities using $100$ independent
simulation run.
\begin{figure}[t!]
\centering
\includegraphics[width=0.8\columnwidth]{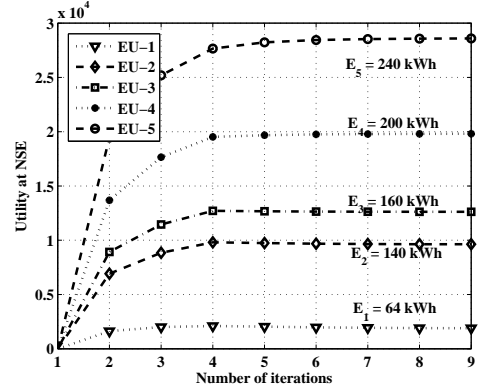}
\caption{Convergence of the utility to NSE.} \label{fig-utility-SE}
\end{figure}
\begin{figure}[t!]
\centering
\includegraphics[width=0.8\columnwidth]{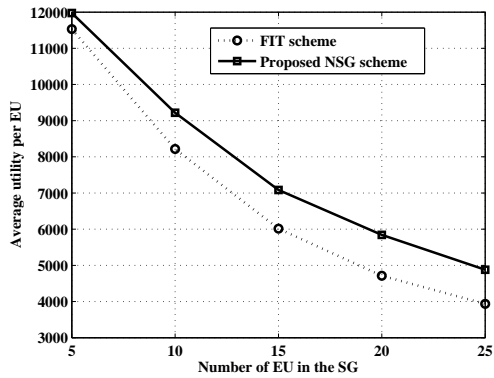}
\caption{Effect of the network size on the average utility per EU.}
\label{fig-utilitywith-EU}
\end{figure}
\begin{figure}[t!]
\centering
\includegraphics[width=0.8\columnwidth]{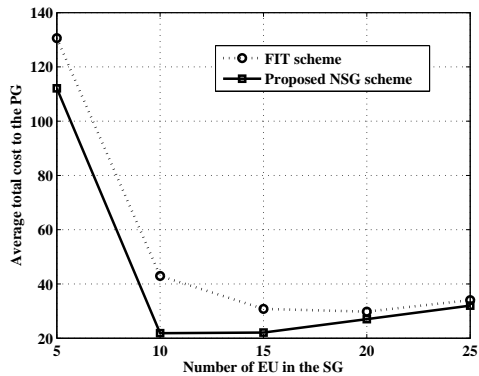}
\caption{Effect of the network size on the average total cost of the
PG.} \label{fig-costwith-EU}
\end{figure}

In Fig.~\ref{fig-utility-SE}, the convergence of the utility
achieved by each EU from selling its energy is shown to reach the
NSE. We consider that $5$ EUs are connected to the PG, and as shown
in the figure, the utility achieved by each EU reaches its NSE after
the $7^{\text{th}}$ iteration. Importantly, the EU with a larger
amount of SE has a higher utility, which is due to the fact that it
can sell more energy to the PG, and hence, it is being paid more.
Consequently, its utility is larger.

Next, we demonstrate the effectiveness of the proposed scheme by
comparing its performance with a standard feed-in tariff scheme
(FIT)~\cite{2012-solarchoice}. We note that an FIT is an incentive
based energy trading scheme which is designed to increase the use of
renewable energy systems providing power to the main grid when it is
required. A higher tariff is paid to the energy producers to
encourage them to take part in energy trading. For comparison, here
we assume that the contract between the EUs and the PG is such that
the EUs are capable of providing the PG with the required energy
with a tariff of $60$ cents per kWh~\cite{2012-solarchoice}. In
Fig.~\ref{fig-utilitywith-EU}, we show the performance comparison
between the proposed and FIT schemes for the average achieved
utility per EU as the number of EUs varies in the SG. As shown in
this figure, an increase in the number of EUs subsequently increases
the freedom of the PG to buy its energy from more EUs, and hence,
the amount of energy sold by each EU decreases. As a result, average
utility decreases for both the schemes. However, the proposed NSG,
due to its capability of choosing an optimal energy for maximizing
the EUs' benefits, always shows a considerable improved performance
over the FIT scheme in terms of average utility per EU. As seen in
Fig.~\ref{fig-utilitywith-EU}, the utility per EU for the proposed
NSG is $1.5$ times, on average, better than the utility achieved by
an FIT scheme.

The effect of the number of EUs on the average total cost to the PG
is shown in Fig.~\ref{fig-costwith-EU} for both the NSG and the FIT
schemes for the same total price per unit of energy $P_r$. For a
fixed $E_{\text{PG}}$, increasing the number of EUs from $5$ to $15$
allows the PG to buy its energy from more EUs, and thus, enables the
PG to pay a cheaper rate. Consequently, the total cost incurred by
the PG decreases. However, to keep all the EUs participating, the PG
needs to pay the minimum mandatory price $p^{\text{min}}$ to each
EU. Thus, as the number of EUs increases from $20$ to $25$, the
total cost to the PG increases due to the mandatory payment to a
large number of EUs. Fig.~\ref{fig-costwith-EU} shows that the
proposed scheme has significantly lower total cost to the PG at
small network sizes, e.g., for $10$ EUs the average total cost for
the proposed scheme is half the total cost incurred by the FIT
scheme. However, as the network size increases, the average total
cost for the proposed NSG becomes closer to the FIT scheme.  In
fact, as the network size increases, the PG needs to optimize its
price for a large number of EUs while maintaining the minimum
payment. Hence, due to the constraint \eqref{constraint-2}, a large
number of EUs causes the PG to choose a price close to its minimum
payment and consequently, the total cost for the proposed NSG
becomes closer to that of the FIT scheme.

\section{Conclusion}\label{conclusion}
In this paper, we have studied a demand-supply balance technique by
prioritizing consumer benefits, and have proposed a Stackelberg
game which leads to a socially optimal Stackelberg equilibrium. We have
shown that the proposed scheme maximizes  the utility of the end
users at the solution of the game, and at the same time keeps the
total cost to the power grid to a minimum. We have
studied the properties of the game analytically including the existence and the
social optimality of the studied scenario. The effectiveness of the
scheme has been demonstrated with considerable performance
improvement when compared to a standard feed-in tariff scheme.\vspace{-0.46cm}

\end{document}